\documentclass[draftcls,onecolumn]{IEEEtran}

\usepackage{amsthm,amsmath,amssymb}

\usepackage[disable]{todonotes}
\usepackage[markup]{changes}

\ifCLASSINFOpdf
\else
\fi

\usepackage{color}
\usepackage{hyperref}

\definecolor{darkblue}{rgb}{0.0,0.0,0.3}
\hypersetup{colorlinks,breaklinks,
            linkcolor=darkblue,urlcolor=darkblue,
            anchorcolor=darkblue,citecolor=darkblue}

\usepackage{url}

\newtheorem{mytheorem}{Theorem}

\newtheorem{mylemma}[mytheorem]{Lemma}
\newcommand{\fourier}[1]{\hat{#1}}
\newcommand{\Fourier}{\mathcal{F}}
\newcommand{\spec}{\text{spec}}
\newcommand{\trace}{\text{tr}\,}
\newcommand{\schattenclass}{\mathcal{I}}
\newcommand{\opnorm}{\text{op}}
\newcommand{\Order}{\mathcal{O}}
\newcommand{\Reals}{\mathbb{R}}
\newcommand{\Naturals}{\mathbb{N}}
\newcommand{\traceP}{\text{tr}_\alpha\,}
\renewcommand{\Pr}{\text{Pr}}
\newcommand{\pw}{S}

\renewcommand{\url}[1]{}

\newcommand{\hbesov}{\dot{\text{B}}}
\newcommand{\besov}{\text{B}}

\begin{document}
%
\title{The Szeg\"o--Asymptotics for Doubly--Dispersive Gaussian Channels}

\author{\IEEEauthorblockN{Peter Jung}
  \IEEEauthorblockA{
    TU Berlin,
    Einsteinufer 25, 10587 Berlin, Germany \\
    peter.jung@tu-berlin.de}
}


%


\maketitle

\begin{abstract}
We consider the time--continuous doubly--dispersive channel with additive Gaussian noise
and establish
a capacity formula for the case where the channel operator is represented by a 
symbol which is periodic in time and fulfills some further
integrability, smoothness and oscillation conditions.
More precisely, we apply the well--known Holsinger-Gallager model for
translating a time--continuous channel for a sequence of
time--intervals of increasing length $\alpha\rightarrow\infty$ to a
series of equivalent sets of discrete, parallel channels, known at the
transmitter. 
We quantify conditions when this procedure converges.
Finally, under periodicity assumptions this result can indeed 
be justified as the channel capacity in the sense Shannon.
The key to this is result is a new Szeg\"o formula for certain
pseudo--differential operators with real--valued symbol. 
The Szeg\"o limit holds if the symbol belongs to the homogeneous Besov space
$\hbesov^1_{\infty,1}$ with respect to its time--dependency,
characterizing the oscillatory behavior in time.
Finally, the formula justifies the water--filling principle in time
and frequency as general technique independent of a sampling scheme.

\end{abstract}


%
\IEEEpeerreviewmaketitle

\let\thefootnote\relax\footnotetext{
  * This work has been partially
  presented on the IEEE ISIT conference, 2011 \cite{jung:isit2011}.
}

\section{Introduction}
The information--theoretic treatment of the time--continuous  channel 
dispersive in time and frequency (doubly--dispersive) with additive
Gaussian noise has been a problem of long interest. 
A well known result for the time--invariant
and power--limited case has been achieved 
by Gallager and Holsinger \cite{Holsinger1964} and \cite{gallager:inftheo} in 
discretizing  the time--continuous problem into an increasing sequence of parallel memoryless channels
with known information capacity $I_n$. Coding theorems for the time--discrete Gaussian channel 
can be used for the time--continuous channel whenever such a discretization is realizable.
A direct coding theorem without discretization has been established by Kadota and Wyner \cite{Kadota1972}
for the causal, stationary and asymptotically memoryless channel. 

The discretization in \cite{gallager:inftheo}
was achieved by representing a single use of the time--continuous channel as the restriction of the channel
operator to time intervals $\alpha\Omega$ of length $\alpha$. The quantity $I_n$ is then determined by 
spectral properties of the restricted operator.
A major step in the calculation for the time--invariant case was the exact determination of the limit:
\begin{equation}
   I(\pw):=\lim_{\alpha\rightarrow\infty}
   \left(\frac{1}{\alpha}\lim_{n\rightarrow\infty}I_n(\alpha \pw)\right)
   \label{eq:capacity:contin}
\end{equation}
which relies on the Kac--Murdock--Szeg\"o result \cite{Kac53} on the asymptotic spectral behaviour of
convolution operators. As the classical result of Shannon for the time--continuous band--limited 
channel and the discussion in \cite{wyner:bandlimited:capacity} shows, $I(S)$ has only 
a meaning of coding capacity for a power budget $\pw$ whenever there exists a sequence $\alpha_k$ of 
realizable discretization approaching this limit as $k\rightarrow\infty$. Some remaining problems in this direction, like for example the 
robustness of this limit against interference between 
different blocks, have been resolved for Gallager--Holsinger model
in \cite{cordaro:intersymbol:gaussian}. The limit has the advantage
of nice interpretation as ''water--filling'' along the frequencies:
\begin{equation}
   I(\pw)=\int_{B\cdot\sigma(\omega)\geq 1}\log(B\cdot\sigma(\omega)) d\omega
   \label{eq:capacity:holsinger}
\end{equation}
where $\sigma$ denotes Fourier transform of the correlation operator $L_\sigma$
(required to be absolute integrable and bounded), i.e. the positive symbol of 
a convolution operator. The constant $B$ is implicitly determined for a given power budget
by a relation similar to \eqref{eq:capacity:holsinger}.  

Since the time--invariant case represents the commutative setting a joint signaling scheme (like for
example orthogonal frequency division multiplexing) is 
permitted and the determination of the capacity is essentially reduced to a power allocation problem. Although the
coherent setting is considered so far  only the channel gains have to be given to the transmitter in this case.

However, doubly--dispersive channels represent the non--commutative generalization and do not admit a joint 
diagonalization such that there still remains the problem of proper signal design. Here, the 
channel operator can be characterized for example by the time--varying transfer function, i.e. the symbol 
$\sigma(x,\omega)$
of a so called pseudo-differential operator which depends on the frequency $\omega$ and the time instant $x$.
Obviously, by uncertainty an exact characterization of frequencies at time instants is meaningless and
the symbol can reflect spectral properties only in an averaged sense. Thus, it is important to know
whether the limit in \eqref{eq:capacity:contin} for a real symbol is asymptotically given by the average:
\begin{equation}
   \frac{1}{\alpha}\iint_{\alpha\Omega\times\Reals}r(B\cdot\sigma(x,\omega))dxd\omega
   \label{eq:capacity:ltv}
\end{equation}
for $\alpha\rightarrow\infty$ and
$r(x)=\log(x)\cdot\chi_{[1,\infty)}(x)$. Then, \eqref{eq:capacity:ltv}
with a similar integral for the function
$(x-1)/x\cdot\chi_{[1,\infty)}(x)$ represents the water--filling
principle in time and frequency. Obviously, this strategy is used
already in practice when optimizing rate functions in some long--term
meaning. But, in fast--fading scenarios for example it not clear
whether this procedure on a short time scale is indeed related to
\eqref{eq:capacity:contin}.

Averages closely related to the one in \eqref{eq:capacity:ltv} have been studied for a long time
in the context of asymptotic symbol calculus of pseudo-differential operators and
semi--classical analysis in quantum physics \cite{Widom1982,Hormander1979,Sobolev2010}. Unfortunately, the results
therein are not directly applicable in the information and communication theoretic setting because here
1.) the symbols of the restricted operators are (in general) discontinuous and usually not decaying in 
time 2.) the functions $r$ to be considered are neither analytic nor have the required smoothness
3.) the path of approaching the limit has to be explicitly in terms of an increasing sequence 
of interval restrictions (infinite--dimensional subspaces) in order to establish realisability.
For operators with semigroup property as the ''heat channel'' \cite{Hammerich2009}
it is possible approaching the limit via projections onto the (finite--dimensional)
span of an increasing sequence of Hermite functions as established in \cite{Zelditch1989} for 
Schr\"odinger operators. However, in the problem considered here this approach does not guarantees
the existence of signaling schemes of finite length $\alpha$ to practically achieve the limit
and a semi--group property of this particular type is not present.

\todo{Neuere arbeit \cite{Hammerich:2015} diskutieren}

The idea of approximate eigenfunctions of so called underspread channels \cite{kozek:thesis,jung:wcnc08}
has been used to obtain information--theoretical statements for
the non--coherent setting \cite{Durisi2010}. Signal design has then to be considered 
with respect to statistical properties \cite{jung:wssuspulseshaping}.
The method presented in this paper suggests that in the 
coherent setting the approximation in terms of trace norms is relevant.

\subsection{Main Results}
We establish a procedure for estimating 
the deviation of formula \eqref{eq:capacity:ltv} to desired quantity \eqref{eq:capacity:contin}.
It will be shown that both terms asymptotically agree for $\alpha\rightarrow\infty$
if the difference of symbol products $L_{\sigma\tau}$ and operator composition
$L_\sigma L_\tau$ can be controlled in trace norm on $\alpha\Omega$ with a sub-linear scaling in $\alpha$.
We will further discuss the information--theoretical impacts:
\begin{mytheorem}
   Let be $\sigma\in C^3$ be the symbol of 
   the channel's correlation operator $L_\sigma$ and $\sigma(x,\cdot)\in L_2$ uniformly in $x$.
   If $\lVert\sigma(x,\cdot)-\sigma(x,\cdot+h)\rVert^2_{L_2}\leq c |h|^{\beta}$ for $\beta\leq 1$
   and $\sigma(x,\omega)$ is $1$--periodic in $x$ the time--continuous capacity under 
   an average power constraint $S$ is given as:
   \begin{equation}
     I(S)=\iint_{(\Omega\times\Reals)\cap\{B\sigma\geq1\}}\log(B\cdot\sigma(x,\omega))dxd\omega
      \label{eq:capacity:ltv:periodic}
   \end{equation}
   with the constant $B=B(S)$ implicitly given by the equation:
    \begin{equation}
       S=\iint_{(\Omega\times\Reals)\cap\{B\sigma\geq1\}}\frac{B\cdot\sigma(x,\omega)-1}{B\cdot\sigma(x,\omega)}dxd\omega
      \label{eq:thm:power:ltv:periodic}
   \end{equation}
   if the (inverse) Fourier transform of $\sigma(x,\omega)$ in $\omega$ (the impulse response of 
   $L_\sigma$) is supported in a fixed interval.
   \label{thm:main}
\end{mytheorem}
The paper is organized as follows: In Section \ref{sec:sysmodel} we introduce the channel model
and establish the problem as a Szeg\"o statement on the asymptotic symbol calculus for
pseudo--differential operators. The asymptotic is investigated in Section \ref{sec:asympt:trace}
as series of four sub--problems: an increasing family of interval sections, the asymptotic symbol calculus, 
an approximation method and finally a result on ''products'' of symbols. Following this line of four
arguments we able establish \eqref{eq:capacity:ltv:periodic}.

\section{System Model and Problem Statement}
\label{sec:sysmodel}
We use $L_p(\Omega)$ for usual Lebesgue spaces ($1\leq p\leq\infty$) of complex--valued functions 
on $\Omega\subseteq\Reals^n$ and abbreviate $L_p=L_p(\Reals^n)$ with corresponding norms 
$\lVert\cdot\rVert_{L_p}$. For $p=2$ the Hilbert space has inner product
$\langle u,v\rangle:=\int \bar{u}v$.
Classes of smooth functions up to order $k$ are denoted with $C^k$ and $\fourier{f}=\Fourier f$ 
is the Fourier transform of $f$. Partial derivatives of a function $\sigma(x,\omega)$ are written
as $\sigma_x$ and $\sigma_\omega$, respectively. $\schattenclass_2$ and $\schattenclass_1$ are
Hilbert--Schmidt and trace class operators with square--summable and absolute summable 
singular values and the symbol $\trace(X)$ denotes the trace of
an operator $X$ (more details will be given later on) on $L_2$.
\subsection{System Model}

We consider the common model of transmitting a finite energy signal $s$ with support in an 
interval $\Omega$ of length $\alpha$ through a channel represented by
a fixed linear operator $H$ and additive distortion $n_k$, i.e. 
quantities simultaneously measured at the receiver within the interval are expressed as noisy 
correlation responses:
\begin{equation}
   \langle r_k,Hs\rangle+n_k
   \label{eq:tc:channeluse}
\end{equation}
where $\{\langle r_k,\cdot\rangle\}$ are suitable normalized linear functionals implemented at the receiver. 
We assume Gaussian noise with $E(\bar{n}_kn_l)=\langle r_k,r_l\rangle$.

Let us denote with $(Pu)(x)=\chi(x/\alpha)u(x)$ the restriction of a 
function $u$ onto the interval $\alpha\Omega$.
Note that in what follows: \emph{$P$ always depends on $\alpha$}.
We will make in the following the assumption that the restriction $HP$ of the channel operator $H$ to 
input signals of length $\alpha$ with finite energy
is compact, i.e. the restriction $PL_\sigma P$ of the correlation operator $L_\sigma:=H^*H$ is compact as well
($H^*$ denotes the adjoint operator on $L_2$).
This excludes certain channel operators - like the identity -
which usually referred to as ''dimension-unlimited'', i.e. the wideband cases.
Assume that the kernel $k(x,y)$ of $L_\sigma$ fulfils:
\begin{equation}
   |k(x,x-z)|^2\leq\psi(z)
\end{equation}
for some  $\sqrt{\psi}\in L_1\cap L_2$\footnote{$\sup_{x\in\Reals} k(x,x-\cdot)\in L_1\cap L_2$}. 
Then its  (Kohn--Nirenberg) symbol or time-varying transfer function
is given by Fourier transformation:
\begin{equation}
   \begin{split}
      \sigma(x,\omega)
      &=\int e^{i2\pi\omega (x-y)} k(x,x-y)dy
   \end{split}
\end{equation}
The symbol is continuous and from the Riemann--Lebesgue Lemma it follows that 
$|\sigma(x,\omega)|\rightarrow 0$ as $|\omega|\rightarrow\infty$.
Throughout the paper we assume that $\sigma$ is real--valued (this can be circumvented 
when passing to the Weyl symbol since $L_\sigma$ is positive--definite). 
It follows that $\lVert \sigma(x,\cdot)\rVert_{L_2}^2=\lVert k(x,x-\cdot)\rVert_{L_2}^2\leq \lVert\psi\rVert_{L_1}$ 
uniformly in $x$ and that $L_\sigma$ is bounded on $L_2$:
\begin{equation}
   \begin{split}
      |\langle u, L_\sigma v\rangle|
      &=|\langle u\otimes\bar{v},k\rangle|
      \leq\langle |u\otimes v|,\sqrt{\psi}\rangle\\
      &=\langle |u|,\sqrt{\psi}\ast |v|\rangle
      \leq\lVert\sqrt{\psi}\rVert_{L_1}\lVert u\rVert_{L_2}\lVert v\rVert_{L_2}
   \end{split}
   \label{eq:szego:boundedOp}
\end{equation}
We will from now on use $\lVert\cdot\rVert_{\opnorm}:=\lVert\cdot\rVert_{L_2\rightarrow L_2}$ 
to denote the operator norm on $L_2$.
A compact operator $HP$ can be written via the Schmidt representation (singular value decomposition) as a
 limit of a sum of
rank--one operators $HP=\sum_k s_k\langle u_k,\cdot\rangle v_k$ with singular values 
$s_k=\sqrt{\lambda_k(PL_\sigma P)}$ and orthonormal bases
$\{u_k\}$ and $\{v_k\}$ -- \emph{all depending on $\alpha$}. 
For the coherent setting we assume that finite subsets of these 
bases are known and implementable at the transmitter and the receiver, respectively. 
Obviously, this is an idealized and seriously strong assumption which can certainly 
not be fulfilled without error in practise. The investigations in \cite{kozek:identification:bandlimited} suggest
that underspreadness of $H$ is necessary prerequisite for reliable error control.
When representing the signal $s$ as a finite linear combination of $\{u_k\}$ 
a single use of the time--continuous channel $H$ over the time interval  
$\alpha\Omega$ with power budget $\pw$ is decomposed into a single use of a finite set of parallel 
Gaussian channels jointly constrained to $\alpha \pw$.

We will consider in the following independent uses of the channel in \eqref{eq:tc:channeluse} as our 
preliminary model and restrict to $r_k=v_k$, i.e. $E(\bar{n}_kn_l)=\delta_{kl}$. Then, the capacity and the power budget 
of the equivalent memoryless Gaussian channel are related through the water--filling level $B$ 
as (see for example \cite{gallager:inftheo}):
\begin{equation}
   \begin{split}
      \frac{1}{\alpha}\sum_{B\lambda_k\geq 1}\log(B\lambda_k)
      =\frac{1}{\alpha}\traceP r(B\, PL_\sigma P)\\
      \frac{B}{\alpha}\sum_{B\lambda_k\geq 1}\frac{B\lambda_k-1}{B\lambda_k}
      =\frac{B}{\alpha}\traceP p(B\,  PL_\sigma P)\\
   \end{split}
   \label{eq:capacity:restricted}
\end{equation}
with $r(x)=\log(x)\cdot\chi_{[1,\infty)}(x)$ and $p(x)=\frac{x-1}{x}\cdot\chi_{[1,\infty)}(x)$.
The symbol $\traceP Y:=\trace(PYP)$ denotes the trace of the operator $Y$ on the range of $P$ and 
the operators $r(PXP)$ and $p(PXP)$ for $X$ being self--adjoint
are meant by spectral mapping theorem.

If the time--varying impulse response of $L_\sigma$ (or $H$)
has finite delay ($k(x,x-z)$ is zero for $z$ outside a fixed interval) and is periodic in
the time instants $x$ (the symbol $\sigma(x,\omega)$ is periodic in $x$) multiple channel 
uses in the preliminary model can be taken as consecutive uses
of the same time--continuous channel. Inserting guard periods of appropriate fixed 
(independent of $\alpha$) size will not affect the asymptotic for $\alpha\rightarrow\infty$.
Thus, any further results will then
indeed refer to the information (and coding) capacity. The assumptions on finite delay 
might be relaxed using directs methods like in \cite{cordaro:intersymbol:gaussian} or 
\cite{wyner:intersymbol:gaussian} whereby extensions to almost--periodic channels seems
to lie at the heart of information theory.

\subsection{Problem Statement}
The interval restriction $P$ has the symbol $\chi(\cdot/\alpha)$. The symbol 
of operator products is given as the twisted multiplication of the symbol of the factors. Under the trace
this is reduced to ordinary multiplication (see for example \cite{Estrada1989} in the case of Weyl correspondence).
Thus, the term in \eqref{eq:capacity:ltv} can be written as the following trace:
\begin{equation}
   \frac{1}{\alpha}\traceP L_{f(\sigma)}=\frac{1}{\alpha}\int_{\alpha\Omega\times\Reals} f(\sigma(x,\omega))dxd\omega
   \label{eq:capacity::ltvapprox}
\end{equation}
when taking $f(x)=r(Bx)$. Comparing \eqref{eq:capacity:restricted} with \eqref{eq:capacity::ltvapprox} means
to estimate the asymptotic behavior of:
\begin{equation}
   \begin{split}
      \frac{1}{\alpha}\traceP(f(P L_\sigma P)-L_{f(\sigma)})
   \end{split}
   \label{eq:capacity:error}
\end{equation}
for $\alpha\rightarrow\infty$ (we abbreviate $f(\sigma):=f\circ\sigma$).
As seen from $r$ and $p$ in \eqref{eq:capacity:restricted} the functions 
$f$ of interest are continuous but not differentiable at $x=1$. 

\section{Asymptotic Trace Formulas}
\label{sec:asympt:trace}
The procedure for estimating the difference in \eqref{eq:capacity:error} essentially consists
in the following arguments:
A functional calculus will be used to represent the function $f$ in the operator context.
For $L_{f(\sigma)}$ this can be done independently of $\alpha$ but 
for $f(P L_\sigma P)$ such an approach is much more complicated because of the 
remaining projections $P$. Hence, the first
step is to estimate its deviation to $f(L_\sigma)$ by inserting the zero term 
$\traceP(f(L_\sigma)-f(L_\sigma))/\alpha$ into \eqref{eq:capacity:error}:
\begin{equation}
   \begin{split}
      \frac{1}{\alpha}(\,\overbrace{\traceP [f(PL_\sigma P)-f(L_\sigma)]}^{\text{stability}}+
      \overbrace{\traceP [f(L_\sigma)-L_{f(\sigma)}]}^{\text{symbol calculus}}\,)
   \end{split}
   \label{eq:szego:splitup}
\end{equation}
and use $|\trace(a+b)|\leq |\trace a |+|\trace b |$ to estimate both terms separately. The
first contribution refers to the stability of interval sections (in Section \ref{sec:szego:stability}).
For second term a Fourier--based functional calculus reduces the problem
to the characterization of the approximate product rule for symbols (in Section \ref{sec:szego:asympt:symbol})
which can then be estimated independently of the particular function $f$ (in
Section \ref{sec:szego:approx:product}).
Unfortunately, the last steps require certain smoothness of $f$. Therefore we will approach the 
limit via smooth approximations $f_\epsilon$ as discussed in Section \ref{sec:szego:smooth:approx}.

\subsection{Stability of Interval Sections}
\label{sec:szego:stability}
The following stability result was inspired by the analysis on the Widom conjecture in \cite{Gioev2001}.
Let $\spec(L_\sigma)$ denote
the spectrum of $L_\sigma$. Then the interval $ I:=\bigcup_{t\in[0,1]} t\cdot \spec(L_\sigma)$
contains the spectra of the family $P L_\sigma P$ for each $\alpha$.
\begin{mytheorem}
   Let $L_\sigma$ be an operator with a kernel which fulfils $|k(x,x-z)|^2\leq \psi(z)$ with 
   $\psi\in L_1$. If $\lVert \psi(1-\chi_{[-s,s]})\rVert_{L_1}\leq c/s$
   then:
   \begin{equation}
      \frac{1}{\alpha}|\traceP (f(P L_\sigma P)-f(L_\sigma))|\leq 
      \lVert f''\rVert_{L_\infty(I)}\frac{\log(\alpha)}{\alpha}
   \end{equation}
   for $f''\in L_\infty(I)$.
   \label{thm:szego:laptev}      
\end{mytheorem}
Further details, see here \cite{Laptev1996:Berezin}.
Recall that the functions $f$ to be considered here
are continuous and differentiable a.e. on $I$ (except at point $x=1$).
We will shortly discuss the proof of this theorem since it is only a minor variation of
\cite{Gioev2001}.
\begin{proof}
   Laptev and Safarov \cite{Laptev1996} have obtained from Berezin inequality the following
   estimate. For functions $f''\in L_\infty(I)$ the operator
   $P[f(L_\sigma ) -  f(PL_\sigma P)]P$ is trace class
   if $PL_\sigma $ and $PL_\sigma (1-P)$ are Hilbert--Schmidt with the trace estimate:
   \begin{equation}
      |\traceP (f(L_\sigma )- f(PL_\sigma P))| \leq\frac{1}{2}\lVert f''\rVert_{L_\infty(I)}
      \lVert PL_\sigma (1-P)\rVert^2_{\schattenclass_2}
   \end{equation}
   Recall that the interval projection $P$ is multiplication with the scaled characteristic function $\chi(\cdot/\alpha)$. Thus, change of variables $x=y'+x'$ and $y=y'-x'$ gives:
   \begin{equation}
      \begin{split}
         \lVert PL_\sigma \rVert_{\schattenclass_2}^2
         &=\int\chi(x/\alpha)|k(x,y)|^2  dxdy\\
         &\hspace*{-2em}\leq2\alpha^2\int\psi(2\alpha x')dx'\int\chi(y'+x')dy'\leq\alpha \lVert\psi\rVert_1
      \end{split}
   \end{equation}
   In the same manner we get:
   \begin{equation}
      \begin{split} 
         \lVert PL_\sigma (1&-P)\rVert^2_{\schattenclass_2}
         =\int \chi(\tfrac{x}{\alpha})(1-\chi(\tfrac{y}{\alpha}))|k(x,y)|^2dxdy\\
         &\leq \alpha^2\int \chi(x)(1-\chi(y))\psi(\alpha (x-y))dxdy\\
         &=\alpha^2\int\psi(2\alpha x)\cdot\omega(2x) dx 
      \end{split}
      \label{eq:thm:szego:laptev:eq2}
   \end{equation}
   with $\omega(x):=4|x|\leq 2$ for $|x|\leq 1/2$ and $\omega(x):=2$ outside this interval.
   With $u=2\alpha x$ and $\phi(u)=\psi(u)+\psi(-u)$ we split and estimate the integral as follows:
   \begin{equation}
      \begin{split} 
         \lVert PL_\sigma (1&-P)\rVert^2_{\schattenclass_2}    
         =\frac{\alpha}{2}\int_0^\infty \phi(u)\omega(\tfrac{u}{\alpha})du\\
         &\leq\frac{\alpha}{2}\left(\tfrac{8}{\alpha}\int_0^2
           +\int_2^{2\alpha}\tfrac{4u}{\alpha}+2\int_{2\alpha}^\infty \right)\phi(u)du
      \end{split}
   \end{equation}  
   With the assumptions of the theorem the terms are bounded separately: 
   \begin{equation}
      \begin{split} 
         \lVert PL_\sigma (1-P)\rVert^2_{\schattenclass_2}    
         &=4\lVert \psi\rVert_1
         +2\int_2^{2\alpha}\phi(u)udu
         +\frac{c}{2}
      \end{split}
   \end{equation}
   Finally we use $\phi(u)=-\frac{d}{du}\int_{u}^{\infty}\phi(s)ds$ and integrate by parts to obtain
   $\int_2^{2\alpha}\phi(u)udu= c (1+\log\alpha)$.
\end{proof}

\noindent{\bf Discussion:}
Consider again the impulse response $h(x,z)=k(x,x-z)$. The condition in the theorem is then:
\begin{equation}
   \lVert h(x,\cdot)(1-\chi_{[-s,s]})\rVert_{L_2}^2
   =\int_{-\infty}^{-s}\left(|h(x,z)|^2+|h(x,-z)|^2\right)dz\leq c/s
\end{equation}
for $s>0$. Differentiating both sides ($d/d(-s)$) gives the sufficient condition:
\begin{equation}
   |h(x,s)|^2+|h(x,-s)|^2\leq c/s^2
\end{equation}
Thus, if the kernel has the decay $|k(x,x-s)|\leq c/|s|$ for any $s\neq 0$ the condition is fulfilled, i.e. from integration
by parts for symbols $\sigma(x,\cdot)\in C^1$ and $\sigma_\omega(x,\cdot)\in L_1$  uniformly in $x$ and vanishing for $\omega$  at infinity.
More generally this holds for $\sigma(x,\cdot)\in L_2$ having $L_2$--modulus of continuity uniformly in $x$ 
$\lVert\sigma(x,\cdot)-\sigma(x,\cdot+h)\rVert^2_{L_2}\leq c |h|^{\beta}$ for $\beta\leq 1$ (see
\cite[Lemma 2.10]{Brandolini1997} for $\beta=1$ and its generalization \cite[Lemma 3.4.1]{Gioev2001} for $0<\beta\leq 1$).

\subsection{Asymptotic Symbol Calculus}
\label{sec:szego:asympt:symbol}
Here we shall use Fourier techniques to estimate 
the right term in \eqref{eq:szego:splitup}.
We abbreviate in the following $e(x)=\exp(i2\pi x)$. 
\begin{mylemma}
   Let $f$ be a $L_1$-function with $\fourier{f}(\omega)=\Order(\omega^{-4-\delta})$
   for some $\delta>0$ and $f(0)=0$.
   For $L_\sigma$ being bounded and self--adjoint on $L_2$ with 
   real--valued symbol $\sigma\in C^{3}$
   it follows that:
   \begin{equation}
      \begin{split}
         \frac{1}{\alpha}|\traceP (f(L_\sigma)-L_{f(\sigma)})|   \leq
         2\pi\int d\omega|\fourier{f}(\omega)|\int_0^\omega Q_\alpha(\nu)\frac{d\nu}{\alpha}
      \end{split}
      \label{eq:szego:lemma:asympt:symbol}
   \end{equation}
   with $Q_\alpha(\nu):=\lVert P\left(L_\sigma L_{e(\nu\sigma )}-L_{\sigma e(\nu\sigma )}\right)P\rVert_{\schattenclass_1}$.
   \label{lemma:szego:asympt:symbol}
\end{mylemma}
The lemma shows that whenever the rhs in \eqref{eq:szego:lemma:asympt:symbol} is finite
the asymptotics for $\alpha\rightarrow\infty$
is determined only by $Q_\alpha/\alpha$. The function $Q_\alpha$ essentially compares 
the twisted product of $\sigma$ and $e(\nu\sigma)$ with the ordinary product $\sigma\cdot e(\nu\sigma )$
in trace norm reduced to intervals of length $\alpha$.
\begin{proof}
  The operator $e(\omega L_\sigma)$ depends continuously on $\omega$ and
  could be defined as the usual power series converging in norm since
  $L_\sigma$ is bounded.
  In particular $e(\omega L_\sigma)$ is unitary on $L_2$ ($L_\sigma$
  is self--adjoint) and $\lVert e(\omega L_\sigma)\rVert_{\opnorm}=1$.
  Thus, for $\fourier{f}\in L_1$ the operator--valued  integral:
   \begin{equation}
      f(L_\sigma)=\int e(\omega L_\sigma) \fourier{f}(\omega)d\omega
      \label{eq:szego:thm:gL}
   \end{equation}
   is a Bochner integral and $\lVert f(L_\sigma)\rVert_\opnorm\leq\lVert \fourier{f}\rVert_{1}$.

   Next, we consider the operator  $L_{f(\sigma)}$ with symbol $f(\sigma)=f\circ\sigma$.
   The value of $f\circ\sigma$ at each point can be expressed in terms of $\fourier{f}$. This suggests
   an integral formula similar to \eqref{eq:szego:thm:gL}:
   \begin{equation}
      L_{f(\sigma)}=\int L_{e(\omega\sigma)}\fourier{f}(\omega)d\omega
      \label{eq:szego:thm:Lg}
   \end{equation}
   where its convergence has to be discussed. Ensuring convergence in
   the sense of Bochner requires further control of  $\lVert L_{e(\nu\sigma)}\rVert_\opnorm$.
   From Calderon Vaillancourt Theorem \cite[Ch.5]{folland:harmonics:phasespace} 
   we have:
   \begin{equation}
      \lVert L_{e(\nu\sigma)}\rVert_\opnorm\leq \lVert e(\nu\sigma)\rVert_{C^3}
      :=\sum_{a+b\leq 3} |2\pi \nu|^{a+b}\lVert\partial_x^a\partial_\omega^b\sigma\rVert_\infty
      \label{eq:szego:thm:LeNorm}
   \end{equation}
   Thus, for $\fourier{f}(\omega)=\Order(\omega^{-4-\delta})$ and
   $\delta>0$ also the integral \eqref{eq:szego:thm:Lg}
   converge in the sense of Bochner. 
   From the considerations above we get therefore:
   \begin{equation}
      \begin{split}
         |\traceP (L_{f(\sigma)}-f(L_\sigma))|
         &\leq\int |\fourier{f}(\omega)|\cdot
         |\traceP u(\omega)|d\omega
      \end{split}
   \end{equation}
   with $u(\omega)=e(\omega L_\sigma)-L_{e(\omega\sigma)}$.
   As suggested in \cite{Sobolev2010} the operator $u(\omega)$ fulfils the following 
   identity\footnote{in the case of operators we use 
   $\partial_\omega e(\omega\sigma)=i2\pi\sigma  e(\omega\sigma)$ \cite[Lemma 5.1]{Gohberg1990}.}:
   \begin{equation}
      \begin{split}
         u'(\omega)
         &=i2\pi\left(L_\sigma  u(\omega)+L_\sigma L_{e(\omega\sigma)}-L_{\sigma e(\omega\sigma)}\right)\\
      \end{split}
      \label{eq:szego:cauchy:inhomogen}
   \end{equation}
   i.e. an inhomogenous Cauchy problem with initial condition  $u(0)=0$.
   By the Stones theorem, $\{u(\omega)\}_{\omega\in\Reals}$ is a strongly (norm--) continuous one--parameter family
   of operators on $L_2(\alpha\Omega)$.
   For a function $g_0\in L_2(\alpha\Omega)$ the solution of the homogeneous equation is its unitary evolution:
   \begin{equation}
      g_\omega = e(i2\pi\omega L_\sigma)g_0= \left[e(i2\pi\omega L_\sigma)P\right]g_0
   \end{equation}
   By Duhamel's principle (see for example
   \cite[p.50]{Cazenave1998} for the Banach--space valued case) the solution of \eqref{eq:szego:cauchy:inhomogen}  
   are the operators:
   \begin{equation}
      \begin{split}
         \hspace*{-1em}u(\omega)=\frac{2\pi}{i}\int_0^\omega e((\omega-\nu)L_\sigma P)
         \left(L_\sigma L_{e(\nu\sigma )}-L_{\sigma e(\nu\sigma )}\right)d\nu
      \end{split}
   \end{equation}
   considered on $L_2(\alpha\Omega)$. With $e((\omega-\nu)L_\sigma P)=e((\omega-\nu)L_\sigma )P$ this gives the estimate:
   \begin{equation}
      \begin{split}
         |\traceP u(\omega)|\leq2\pi\int_0^\omega
         \lVert P\left(L_\sigma L_{e(\nu\sigma )}-L_{\sigma e(\nu\sigma )}\right)P\rVert_{\schattenclass_1}d\nu
         =2\pi \int_0^\omega Q_\alpha(\nu)d\nu
      \end{split}
   \end{equation}
   since $\lVert Pe((\omega-\nu)L_\sigma )\rVert_\opnorm\leq 1$. 
\end{proof}
The smoothness assumptions in the theorem can be weakened to $\sigma\in C^{2+\delta}$ and 
$\fourier{f}(\omega)=\Order(\omega^{-3-\delta})$ when using H\"older-Zygmund spaces. 
Furthermore, there is variant in terms of the modulation space
$M^{\infty,1}$ (see \cite[pp.320]{grochenig:gaborbook}).
We expect that these conditions can be further reduced when
using in \eqref{eq:szego:thm:Lg} some weaker convergence in $\traceP$ 
instead of requiring a Bochner integral. The proof of the theorem can also be based
on the Paley--Wiener theorem, i.e. 
$f\rightarrow f(L_\sigma)$ and $f\rightarrow L_{f(\sigma)}$ are operator--valued
distributions of compact support with order at most $3$  and have therefore 
$C^3$ as natural domain (hier nochmal auf die decay condition im theorem eingehen).

\subsection{An Approximation Procedure}
\label{sec:szego:smooth:approx}
\todo[inline]{$J$ einfuehren}
Since $L_\sigma$ is bounded (see \eqref{eq:szego:boundedOp}) the functions $f$ 
will be evaluated only on a finite interval contained in $I$. We consider 
functions $f$  of the form $f(x)=h(x)\cdot\chi_{[1,\infty)}(x)$
with a critical point at $x=1$ (the function $h$ is ``sufficiently nice'' on $I$, e.g., in $C^\infty(I)$).
In general therefore, the first derivative of $f$ will be not continuous at $x=1$
but of bounded variation.
By smooth extension outside the interval $I$ its Fourier transforms
$\fourier{f}(\omega)$ can decay 
only as $\Order(\omega^{-2})$,  
see here for example \cite[Theorem 2.4]{Trefethen1996}, i.e. $f\in L_1\cap\Fourier L_1$.
Unfortunately, this is not sufficient for Lemma \ref{lemma:szego:asympt:symbol}.
Therefore, we replace the 
Heaviside function $\chi_{[1,\infty)}$
in $f$ by a series of smooth approximations $\phi_\epsilon$ as done for example in \cite{Hormander1979}.
Let be $\phi\in C^\infty$ strictly increasing
\todo[inline]{strictly increasing or not ?}
with $\phi(t)=0$ for $t\leq 0$ and $\phi(t)=1$ for
$t\geq 1$. Define $\phi_\epsilon(x)=\phi(\frac{x-1}{\epsilon})$ and consider  
$f_\epsilon=h\phi_\epsilon\in C^\infty_c$ instead of $f$ (again by smooth extension outside the interval $I$):
\begin{equation}
   \begin{split}
      |\fourier{f}_{\epsilon}(\omega)|
      \leq\frac{c'_n|I|}{|2\pi\omega|^{n}}\epsilon^{-n}
   \end{split}
   \label{eq:szego:smoothapprox:fourier:decay}
\end{equation}

\begin{figure}
   \centering\hspace*{-6em}
   \includegraphics[width=1\linewidth]{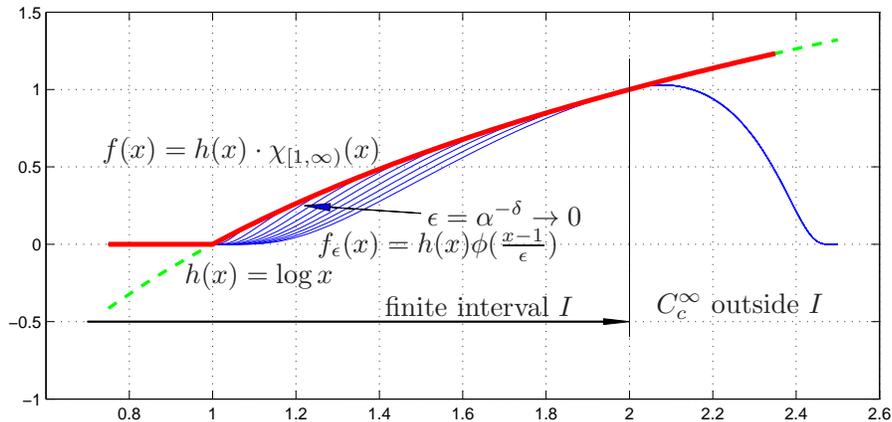}
   \caption{
     Adapting the approximation accuracy to the interval length
     $\alpha$: the symbol calculus has to be applied on
     non--differentiable functions $f(x)$ (red bold) on the interval
     $I=[0,2]$. For example, for the rates $f$ is the
     pointwise product of $h(x)=\log(x)$ and the step--function $\chi_{[1,\infty)}(x)$.
     The latter will be approximated on $I$ by a series of smooth and
     compactly supported
     functions $\phi((x-1)/\epsilon)$ with $\epsilon=\alpha^{-\delta}$
     for some $\delta>0$ such that we have $\epsilon\rightarrow 0$ for $\alpha\rightarrow\infty$.
   }
   \label{fig:heavisideapprox}
\end{figure}

\hrule\vspace*{2em}
We abbreviate $d_\epsilon=f-f_\epsilon$ and obtain from triangle
inequality and linearity in $f$ that:
\begin{equation}
  \begin{split}
    |\traceP (f(L_\sigma)-L_{f(\sigma)})|\leq 
    |\traceP (f_\epsilon(L_\sigma)-L_{f_\epsilon(\sigma)})|+
    |\traceP d_\epsilon(L_\sigma)|+
    |\traceP L_{d_\epsilon(\sigma)}|
  \end{split}
  \label{eq:szego:traceapprox:epsilon}
\end{equation}
We assume wlog that $\max_{t\in J} h(t)=1$.
The function $d_\epsilon$ is non--negative and of the following form:
\begin{equation}
  0\leq d_\epsilon=(\chi_{[1,\infty)}-\phi_\epsilon)h=(\chi_{[1,\infty)}-\phi((\cdot-1)/\epsilon))\cdot
  h\leq 1
\end{equation}
Furthermore,
$d_\epsilon$ is smooth except $x=1$ where it has no continuous
derivative. Therefore $\fourier{d}_{\epsilon}(\omega)$ decays as
$\Order(\omega^{-2})$ implying that $d_\epsilon\in L_1\cap\Fourier L_1$.
On its support $\text{supp}(d_\epsilon)\subseteq[1,1+\epsilon]$
the function  $d_\epsilon$  is upper--bounded by the strictly decreasing
function $g_\epsilon$:
\begin{equation}
  d_\epsilon\leq g_\epsilon:=(1-\phi((\cdot-1)/\epsilon))
\end{equation}

\paragraph{Scaling of $\traceP L_{d_\epsilon(\sigma)}$} 
For the last approximation term in \eqref{eq:szego:traceapprox:epsilon} we compute the trace using
\eqref{eq:capacity::ltvapprox}. Let $\mu$ be Lebesgue--measure in
$\alpha\Omega\times\Reals$ and
$z=(x,\omega)\in\alpha\Omega\times\Reals$. 
Since we assume
$\lVert\sigma\rVert_{L_2(\alpha\Omega\times\Reals)}^2=\Order(\alpha)$
we have the following standard estimate:
\begin{equation}
  \mu\{|\sigma|\geq t\}\leq\frac{1}{t^2}\lVert\sigma\rVert_{L_2(\alpha\Omega\times\Reals)}^2=\Order(\alpha)
\end{equation}
and $\mu\{|\sigma|\geq t\}$ is therefore decreasing in $t$. We also
use the layer cake representation of integrals (see
\cite{lieb:analysis}). Let $\{\nu:F(z)\geq t\}$ be
super--level set of the non--negative measurable function $F$. Then:
\begin{equation}
  F(z)=\int_0^\infty\chi_{\{\nu:F(\nu)\geq t\}}(z)dt
  \label{eq:layercake}
\end{equation}
We start from \eqref{eq:capacity::ltvapprox} with the property that 
$d_\epsilon\leq g_\epsilon\leq 1$ on its support $[1,1+\epsilon]$ and
therefore also $g_\epsilon(\sigma)\leq\sigma=|\sigma|$ on the set where
$1\leq\sigma$. Thus, using Fubini theorem:
\begin{equation}
  \begin{split}
    \traceP L_{d_\epsilon(\sigma)}
    &\leq\int_{\alpha\Omega\times\Reals} g_\epsilon(\sigma(z))d\mu(z)
    \leq\int_{\{1\leq\sigma\leq1+\epsilon\}} |\sigma| d\mu
    \overset{\eqref{eq:layercake}}{=}\int_{\{1\leq\sigma\leq1+\epsilon\}}
    \left(\int_0^\infty\chi_{\{|\sigma|\geq t\}}(z)dt\right)d\mu\\
    &=\int_1^{1+\epsilon}
    \left(\int_{\{1\leq\sigma\leq1+\epsilon\}}\chi_{\{|\sigma|\geq t\}}(z)d\mu\right)dt
    \leq\int_1^{1+\epsilon} \mu\{\sigma\geq t\} dt\leq 
    \epsilon\cdot\mu\{\sigma\geq 1\}=\Order(\epsilon\cdot\alpha)
  \end{split}
\end{equation}
The last inequality follows since $\mu\{\sigma\geq t\}$ is decreasing
in $t$.

\paragraph{Scaling of $\traceP d_\epsilon(L_\sigma)$} 
Let $S=\spec(PL_\sigma P)$ be the spectrum of $PL_\sigma P$. We have the relation (see here \eqref{eq:capacity::ltvapprox}):
\begin{equation}
  \sum_{1\leq \lambda\in S}\lambda\leq \sum_{\lambda\in S} \lambda=\traceP L_{\sigma}=\int_{\alpha\Omega\times\Reals}
  \sigma(x,\omega)dxd\omega=\Order(\alpha)\quad\text{to check}
  \label{eq:szego:trace:epsalpha}
\end{equation}
We perform now similar steps as above for a counting measure instead
of $\mu$ and we abbreviate $S_\epsilon:=\{\lambda\in S:\, 1\leq\lambda\leq1+\epsilon\}$.
From spectral theorem it follows that:
\begin{equation}
  \begin{split}
    \traceP d_\epsilon(L_{\sigma})
    &=\sum_{\lambda\in S}  d_\epsilon(\lambda)
    \leq\sum_{\lambda\in S_\epsilon}g_\epsilon(\lambda)
    \leq\sum_{\lambda\in S_\epsilon}\lambda
    \overset{\eqref{eq:layercake}}{=}\sum_{\lambda\in S_\epsilon}\int_0^\infty\chi_{\{\nu\in S:\,
      \nu\geq t\}}(\lambda)dt\\
    &=\int_0^\infty\left(\sum_{\lambda\in S_\epsilon}\chi_{\{\nu\in S:\,
        \nu\geq t\}}(\lambda)\right)dt
    = \int_0^\infty
    |\{\lambda\in S_\epsilon: \lambda\geq t\}|dt\\
    &\leq \int_1^{1+\epsilon}(\sum_{1\leq \lambda\in S}1)dt
    \leq \epsilon(\sum_{1\leq \lambda\in S}\lambda)\overset{\eqref{eq:szego:trace:epsalpha}}{\leq}\Order(\epsilon\alpha)
  \end{split}
\end{equation}

\hrule
\vspace*{2em}

In essence: polynomial grow of $Q_\alpha(\nu)$ in $\nu$ can always be compensated 
by taking $n$ large enough such that at the rhs in \eqref{eq:szego:lemma:asympt:symbol} remains
a finite quantity $R_\alpha(\epsilon)$. If for example $R_\alpha(\epsilon)=\Order(\alpha^{-\gamma})$, 
we choose $\epsilon=\alpha^{-\delta}$ with $\delta<\gamma/n$. Then 
$R_\alpha(\epsilon)\rightarrow 0$ and $\epsilon\rightarrow 0$ for $\alpha\rightarrow\infty$
which is obviously sufficient for the limit.

\subsection{Approximate Symbol Products}
\label{sec:szego:approx:product}

Polynomial orders of $Q_\alpha(\nu)$ in $\nu$ which will occur in the following will be compensated by the approximation method
in Section \ref{sec:szego:smooth:approx}. The role of $\tau$ and $\sigma$ can also be interchanged 
since according \eqref{eq:szego:thm:LeNorm} $L_\tau$ is bounded polynomially in $s$.
\todo[inline]{ Since in the time--invariant setting this part is not necessary. But Gallager 
(find the paper of Baker ''Information Capacity of Stationary Gaussian Channel'' ...)
requires that $\sigma(\omega)\in L_1$ and we found that $\sigma(\omega)\in L_2$ is enough (why ? due to
the approximation step ?).}

Let us abbreviate $\tau=e(\nu\sigma )=\exp(i2\pi\nu\sigma)$. 
Then the operator in the term $Q_\alpha(\nu)/\alpha$
of Lemma \ref{lemma:szego:asympt:symbol} is the deviation between operator and symbol product
$L_\sigma L_{\tau}-L_{\sigma \tau}$. 
As in \cite{Widom1982} we insert $L_\sigma L_{\bar{\tau}}^*-L_\sigma L_{\bar{\tau}}^*=0$.
Define the operators $T=L_{\bar{\tau}}^*-L_{\tau}$ and $T'=L_\sigma L_{\bar{\tau}}^*-L_{\sigma \tau}$
and apply triangle inequality to obtain:
\begin{equation}
   \begin{split}
      Q_\alpha(\nu)
      &\leq \lVert P L_\sigma TP\rVert_{\schattenclass_1}+\lVert PT'P\rVert_{\schattenclass_1}\\
      &\leq \lVert P L_\sigma PTP\rVert_{\schattenclass_1}+\lVert P L_\sigma (1-P)TP\rVert_{\schattenclass_1}+\lVert PT'P\rVert_{\schattenclass_1}\\
      &\leq \lVert P L_\sigma PTP\rVert_{\schattenclass_1}+\lVert PT'P\rVert_{\schattenclass_1}+
      \lVert P L_\sigma (1-P)\rVert_{\schattenclass_2}\cdot\lVert TP\rVert_{\schattenclass_2}\\
      &\overset{\eqref{eq:thm:szego:laptev:eq2}}{\leq}
      \lVert PL_\sigma\rVert_{\opnorm}\lVert PTP\rVert_{\schattenclass_1}      
      +\lVert PT'P\rVert_{\schattenclass_1}+c\sqrt{1+\log(\alpha)}\lVert TP\rVert_{\schattenclass_2}\\
   \end{split}
   \label{eq:szego:qalpha}
\end{equation}
where \eqref{eq:thm:szego:laptev:eq2} from the proof of Theorem \ref{thm:szego:laptev} has been used.
The operators $T$ and $T'$ have the
kernels $t(x,y)$ and $t'(x,y)$ defined formally as:
\begin{equation}
   \begin{split}
      t(x,y)&=\int e^{i2\pi(x-y)\omega}(\tau(x,\omega)-\tau(y,\omega))d\omega\\
      t'(x,y)&=\int e^{i2\pi(x-y)\omega}\sigma(x,\omega)(\tau(x,\omega)-\tau(y,\omega))d\omega
   \end{split}
   \label{eq:szego:tkernels}
\end{equation}
The meaning of these integrals has to be discussed.
Note again that $\tau(x,\omega)=\exp(i2\pi\nu\sigma(x,\omega))$ is a pure phase symbol and $\sigma\in C^3$, i.e.
from $\sigma(x,\cdot)\in L_2$ follows that $\sigma(x,\cdot)$ vanishes at infinity for each $x$. We have therefore
for all $x$ and $y$:
\begin{equation}
   2\geq |\tau(x,\omega)-\tau(y,\omega)|\rightarrow 0\quad\text{as}\quad |\omega|\rightarrow\infty
\end{equation}
From integration by parts ($n$ times) we have:
\begin{equation}
   \begin{split}
      t(x,y)&=\int e^{i2\pi(x-y)\omega}\frac{\partial^n_\omega(\tau(x,\omega)-\tau(y,\omega))}{(i2\pi(x-y))^n}d\omega
      =:\int e^{i2\pi(x-y)\omega}t_\omega(x,y)d\omega
      \\
      t'(x,y)&=\int e^{i2\pi(x-y)\omega}\frac{\partial^n_\omega(\sigma(x,\omega)(\tau(x,\omega)-\tau(y,\omega)))}{(i2\pi(x-y))^n}d\omega
      =:\int e^{i2\pi(x-y)\omega}t'_\omega(x,y)d\omega
   \end{split}
   \label{eq:szego:tkernels:byparts}
\end{equation}
Both kernels $t_\omega$ and $t'_\omega$ are finite linear combinations of the form:
\begin{equation}
   \left(\sigma^{(k)}(x,\omega)^j\right)\frac{\partial^n_\omega(\tau(x,\omega)-\tau(y,\omega))}{(i2\pi(x-y))^m}
   \label{eq:szego:tkernels:form}
\end{equation}
for $j=0\dots$ and $m\geq n$ where $n\geq 1$ is fixed. 
\todo[inline]{Really $j=0$ ? Or $j=1$}
We shall argue later that it will be sufficient to consider the case $n=1$.
Let then $T_\omega$ and $T'_\omega$ be the corresponding operators
with the kernels $t_\omega$ and $t'_\omega$.
Once the trace norms of the restricted operators $PT_\omega P$ and
$PT'_\omega P$ decay sufficiently the estimate:
\begin{equation}
   \lVert PTP\rVert_{\schattenclass_1}\leq \int \lVert PT_\omega P\rVert_{\schattenclass_1}  d\omega\quad(\text{same for }T')
   \label{eq:szego:PTP}
\end{equation}
is possible. This will be related to the oscillatory behavior of the
symbol and we will investigate this in the next section.
Before continueing on this point, we consider the Hilbert--Schmidt norm $\lVert TP\rVert_{\schattenclass_2}$
in \eqref{eq:szego:qalpha}. From \eqref{eq:szego:tkernels:byparts} for
$n=1$ we get:
\begin{equation}
   \begin{split}
      t(x,y)&=\frac{h(x,x-y)-h(y,x-y)}{i2\pi(x-y)}
   \end{split}
\end{equation}
where $h(x,z)=\int e^{i2\pi\omega z}\tau_\omega(x,\omega)d\omega$ and 
$|\tau_\omega|=|2\pi s\sigma_\omega|$. 
If $\sigma_{x\omega}(x,\cdot)\in L_2$ uniformly in $x$ we deduce with the mean value theorem
that $|t(x,y)|^2\leq c/(1+|x-y|^2)$. This in turns implies that
$\lVert TP\rVert_{\schattenclass_2}=\Order(\sqrt{\alpha})$. Since $L_\sigma$ is 
bounded \eqref{eq:szego:qalpha} yields:
\begin{equation}
   \begin{split}
      Q_\alpha(\nu)
      \leq
      c_1\lVert PTP\rVert_{\schattenclass_1}      
      +\lVert PT'P\rVert_{\schattenclass_1}+c_2\sqrt{\alpha(1+\log(\alpha))}
   \end{split}
   \label{eq:szego:qalpha1}
\end{equation}
Thus, it remains to evaluate trace norms of kernels $t(x,y)$ and
$t'(x,y)$ restricted to $\alpha\Omega\times\alpha\Omega$
(recall that the restriction operators $P=P_\alpha$ depend on $\alpha$
and the operators $T$, $T'$ depend on $\nu$ through their kernels $t$
and $t'$).

\subsection{Paracommutators and Schur Multipliers}
Under mild assumptions on the $\omega$--dependency of $\sigma(x,\omega)$
the overall trace norm  $\lVert PTP\rVert_{\schattenclass_1}$ and $\lVert PT'P\rVert_{\schattenclass_1}$ in \eqref{eq:szego:qalpha1}
can be estimated following the decomposition in \eqref{eq:szego:PTP}
and the oscillatory character of $\sigma(x,\omega)$ in $x$
will play the major role.
We will express the contributions $PT_\omega P$ and $PT'_\omega P$ in
\eqref{eq:szego:PTP}
in the form of \emph{paracommutators}
which have been studied for example by Janson and Peetre in \cite{Janson1988}. A
paracommutator $T_b$ with \emph{symbol} $b$ is a
``para--multiplication'' of the following form:
\begin{equation}
  \widehat{T_b f}(\xi)=\int
  \fourier{b}(\xi-\eta)A(\xi,\eta)\fourier{f}(\eta)d\eta
  \label{eq:szego:def:paracomm}
\end{equation}
where the function $A(\xi,\eta)$ acts as a Schur multiplier
(see below).
This definition includes Toeplitz and Hankel operators, i.e., for
example for $A(\xi,\eta)=1$ its a pointwise multiplication $f\cdot b$.

\paragraph{Schur Multipliers} An important tool therein are bounded Schur--multipliers
(see here also \cite{Peller1985}). For exposition, let $m(x,y)$ be a function represented by:
\begin{equation}
   m(x,y)=\int m_x(x,t)m_y(y,t)d\mu(t)
   \label{eq:szego:schur:def}
\end{equation}
with some sigma--finite measure $\mu$ and denote with $M$ the corresponding mapping which operates on 
kernels $k$ of operators $K$ of a particular Schatten ideal $\schattenclass_p$ for $1\leq p\leq\infty$. Then:
\begin{equation}
   \lVert MK \rVert_{\schattenclass_p}\leq 
   \lVert K\rVert_{\schattenclass_p}\cdot\int \lVert m_x(\cdot,t)\rVert_{L_\infty}\cdot \lVert m_y(\cdot,t)\rVert_{L_\infty}d\mu(t)
   \leq\lVert K\rVert_{\schattenclass_p}\cdot \lVert M\rVert_M
\end{equation}
and $\lVert M\rVert_M$ is essentially defined here as the infimum over all representations
\eqref{eq:szego:schur:def}.
For $\mu$ being a simple point--measure at $t_0$ the kernel of $MK$
takes the form $m_x(x)k(x,y)m_y(y)$ with $m_x=m_x(\cdot,t_0)$ and $m_y=m_y(\cdot,t_0)$
such that 
$\lVert MK\rVert_{\schattenclass_p}\leq \lVert K\rVert_{\schattenclass_p}\lVert m_x\rVert_{L_\infty}\lVert m_y\rVert_{L_\infty}$.

A direct consequence of these tools is that for investigation of scaling

the term
$\sigma^{(k)}(x,\omega)^j$ in \eqref{eq:szego:tkernels:form} can be
dropped if $\sigma(x,\omega)$ and it derivatives like $\partial_\omega\sigma(x,\omega)$
are uniformly bounded in $x$ and $\omega$.
They are Schur--multipliers in $x$ for fixed $\omega$ with $\lVert\sigma(\cdot,\omega)\rVert_M=\lVert\sigma(\cdot,\omega)\rVert_{L_\infty}\leq C$
where the constant $C$ does not depend on $\omega$ and the size $\alpha$ of interval $\alpha\Omega$.

\paragraph{The Paracommutator}
From Schur--multiplier techniques it follows also that we can replace the outer support restriction $\chi_\alpha(x)\chi_\alpha(y)$
with certain (inner) smooth functions $\phi_\alpha(x):=\phi(x/\alpha)$, $\phi_\alpha(y)$ and $\phi_\alpha((x-y)/2)$ 
where $\phi\in C^{\infty}$ with $\phi(x)=1$ for $|x|\leq 1$ and 
$\phi(x)=0$ for $|x|\geq 2$. 
More precisely, since 
\begin{equation}
   \chi_\alpha(x)\chi_\alpha(y)=\chi_\alpha(x)\chi_\alpha(y)\phi_\alpha(x)\phi_\alpha(y)\phi_\alpha((x-y)/2)
\end{equation}
we can write the kernels (we do not write here $\omega$--dependency explicitely) as:
\begin{equation}
   \begin{split}
      k(x,y)
      &=\left[\tau(x)\phi_\alpha(x)-\tau(y)\phi_\alpha(y)\right]\frac{\phi_\alpha((x-y)/2)}{x-y}\cdot\chi_\alpha(x)\chi_\alpha(y)\\
   \end{split}
\end{equation}
where either $\tau(x)=\exp(i2\pi s\sigma(x,\omega))$ or $\tau(x)=\partial_\omega\exp(i2\pi s\sigma(x,\omega))$.
Since the restriction $\chi_\alpha(x)\chi_\alpha(y)$ is a
Schur--multiplier of norm one, the trace norm is upper--bounded by the trace of the smooth kernel:
\begin{equation}
   \begin{split}
      k(x,y)
      &=\left[\tau(x)\phi_\alpha(x)-\tau(y)\phi_\alpha(y)\right]\frac{\phi_\alpha((x-y)/2)}{x-y}
      =:\left[b(x)-b(y)\right]a(x-y)
   \end{split}
\end{equation}
where $a$ is a distribution defined for $z\neq0$ by $a(z)=\phi(z/(2\alpha))/z$ and $b$ is a smooth function defined by:
\begin{equation}
   b(x)=\begin{cases}
      \phi_\alpha(x)e^{i2\pi s \sigma(x,\omega)} \quad&\text{for case}\,(T)\\
      \phi_\alpha(x)\partial_\omega e^{i2\pi s\sigma(x,\omega)} \quad&\text{for case}\,(T')
   \end{cases}
   \label{eq:szego:b}
\end{equation}
Next, we path to the Fourier kernel and we use here abbreviations from \cite{Janson1988}, i.e. we write: 
\begin{equation}
   k(x,y)=\int_{u+v=1}b(ux+vy)a(x-y)d\mu(u,v)
\end{equation}
with the point measure $\mu(u,v)=\delta(u-1,v)-\delta(u,v-1)$. The Fourier kernel $\fourier{k}=(\Fourier\otimes\Fourier^*)k$ of $T$ 
is the kernel of the operator $\Fourier T\Fourier^{*}$ defined by 
$\langle f,Tg\rangle=\langle\fourier{f},\Fourier T\Fourier^*\fourier{g}\rangle$ since
\begin{equation}
   \langle f\otimes \bar{g},k\rangle
   =\langle \Fourier^*\fourier{f}\otimes\overline{\Fourier^*\fourier{g}},k\rangle
   =\langle (\Fourier^*\otimes\Fourier)\fourier{f}\otimes\bar{\fourier{g}},k\rangle
   =\langle \fourier{f}\otimes\bar{\fourier{g}},(\Fourier\otimes\Fourier^*)k\rangle
\end{equation}
We get therefore (repeating the steps in \cite[p.499]{Janson1988}):
\begin{equation}
   \begin{split}
      \fourier{k}(\xi,\eta)
      &=\iint e^{-i2\pi(x\xi-y\eta)}k(x,y)dxdy\\
      &=\int_{u+v=1} d\mu(u,v)\iint e^{-i2\pi(x\xi-y\eta)}b(ux+vy)a(x-y)dxdy\quad\text{with}\, z=x-y\\
      &=\int_{u+v=1} d\mu(u,v)\iint e^{-i2\pi(z\xi+y(\xi-\eta))}b(uz+y)a(z)dzdy\quad\text{with}\, w=uz+y\\
      &=\int_{u+v=1} d\mu(u,v)\iint e^{-i2\pi(z\xi+(w-uz)(\xi-\eta))}b(w)a(z)dzdw\\
      &=\int_{u+v=1} \fourier{b}(\xi-\eta)\fourier{a}(v\xi+u\eta)\,d\mu(u,v)
      =\fourier{b}(\xi-\eta)\left[\fourier{a}(\eta)-\fourier{a}(\xi)\right]\\
      & =:\fourier{b}(\xi-\eta)A(\xi,\eta)\\
   \end{split}
\end{equation}
According to \eqref{eq:szego:def:paracomm} this is the Fourier kernel
of a paracommutator with
$A(\xi,\eta)=\fourier{a}(\eta)-\fourier{a}(\xi)$. Except of the singularity of $a(z)$ at point $z=0$ the definition of the Fourier transforms 
are not problematic since the functions are $a$ and $b$ are continuous and compactly supported.
\newcommand{\PV}{\text{PV}\!\!\!}

\paragraph{Calderon--Zygmund commutator}
Consider now $a(z)=\phi(z/\alpha)/z^m$ for $m\geq n$ which occurs when integrating by parts $n\geq 1$ times in \eqref{eq:szego:tkernels:byparts}.
The dominating order in $\alpha$ is indeed for $m=1$.
The Fourier integral of $a$ is not absolutely convergent unless $\phi(z)=0$, i.e., has only the meaning of a principal value.
To compute
its Fourier transform we use that $\phi\in C^\infty$ and that $\phi^{(m-1)}(0)=0$ for $m>1$, due to the symmetry of 
$\phi$ around zero. We get (see here also \cite[pp.324]{Folland2009}) with $\psi(z):=e^{-i2\pi\alpha\xi z}\phi(z)$:
\begin{equation}
   \begin{split}
      \fourier{a}(\xi)
      &=\int e^{-i2\pi\xi z}\frac{\phi(z/\alpha)}{z^m}dz
      =\alpha^{1-m}\int \frac{\psi(z)}{z^m}dz
      \overset{(m>1)}{=}\frac{\alpha^{1-m}}{(m-1)!}\int \frac{\psi^{(m-1)}(z)}{z}dz\\
      &=\frac{1}{(m-1)!}\int e^{-i2\pi\alpha\xi z}\frac{\phi^{(m-1)}(z)}{z}dz+\Order(\alpha^{-1})
      \rightarrow\text{sgn}(\xi)+\Order(\alpha^{-1})
   \end{split}
\end{equation}
that the leading term is given by the Hilbert transform of $\psi^{(m-1)}$.

In the exposition above we already sticked to $n=1$, i.e. $k(x,y)$ for $m=1$ is the (smoothly truncated) kernel 
of the so called \emph{Calderon--Zygmund commutator} (commutator of a multiplier $b$ with the Hilbert transform 
having the kernel $a$). The Schatten class properties of such type of operators have been investigated by \cite{Peller1982}
and \cite{Rochberg1982} and are related to the oscillatory characterization of the multiplier $b$.
Here we follow the lines of \cite{Janson1988}.

\subsection{Schatten--Properties of Paracommutars and Besov Spaces}
Recall that we need to evaluate the scaling of trace norms in
\eqref{eq:szego:qalpha1}. The corresponding operators $T$ and $T'$ can
be decomposed into integrals \eqref{eq:szego:PTP} once the symbol
$\sigma(x,\omega)$ has sufficient decay in $\omega$ (wideband case are
therefore excluded). Each contribution to the
integrals is a trace norm of restrictions of the operators $T_\omega$
and $T'_\omega$. As explained above, this can related to trace norms
of  para--commutators, i.e. Fourier integral operators of the form
\eqref{eq:szego:def:paracomm} with symbols 
$b$ as defined \eqref{eq:szego:b}. The most well--known example here
is the Calderon--Zygmund commutator which can be written in terms of
Hankel operator. Peller was the first who observed that
such operators are nuclear (trace--class) if the symbol is in a
particular Besov--Space \cite{Peller79}. 
\paragraph{Homogeneous Besov Spaces}
Let $\{\phi_k\}_{k=\infty}^\infty$ be the
Paley--Littlewood decomposition, i.e., 
$\phi=\phi_0\in C_c^\infty$ with
$\text{supp}(\fourier{\phi})\subset\{\frac{1}{2}\leq|\omega|\leq 2\}$
and
$\fourier{\phi}_k:=\fourier{\phi}(2^{-k}\cdot)$ with:
\begin{equation}
  \sum_k \fourier{\phi}_k(\omega)=1\quad\text{for all}\quad 0\neq \omega\in\Reals
\end{equation}
For $-\infty<s<\infty$ and $0<p,q\leq\infty$,  
the homogeneous Besov spaces\footnote{
  In \cite{Herz1968} the spaces are denoted by $\Lambda_{pq}^s$.
}
of distributions (modulo polynomials) are defined by the quasi--norms:
\begin{equation}
   \lVert f\rVert_{\hbesov^s_{pq}}=\left(\sum_{k\in\mathbb{Z}} \left(2^{ks}\lVert \phi_k\ast f\rVert_p\right)^q\right)^{1/q}
   =\lVert\{2^{ks}\lVert \phi_k\ast f\rVert_p\}_{k\in\mathbb{Z}}\rVert_{\ell_q}
   \label{eq:def:hbesov}
\end{equation}
The spaces are called homogeneous since for the whole $(p,q,s)$--range, given above, it follows
that for all $\alpha>0$ it holds (see for example \cite[Proposition 3.8]{Moussai2012}):
\begin{equation}
  c_1\lVert f\rVert_{\hbesov^s_{pq}}\leq
  \alpha^{-(1/p-s)} \lVert f(\cdot/\alpha)\rVert_{\hbesov^s_{pq}}
  \leq c_2\lVert f\rVert_{\hbesov^s_{pq}}
   \label{eq:hbesov:homogenity}
\end{equation}
with equality for $\alpha=2^{-k}$ for some $k$ (see \cite[Remark 4 on p.239]{Triebel:functionspaces:1},
\cite[Remark 2 on p.94]{Runst1996} and also \cite[Lemma 1.2 on p.288]{Herz1968}).
The Besov spaces $\hbesov^s_{pq}$ for $s<1/p$ or $s=1/p$ with $q=1$ can be regarded as subspaces of tempered
distributions $S'$. Obviously, for $q'\leq q$ there holds the inclusion $\hbesov^s_{pq'}\subseteq\hbesov^s_{pq}$ since
$\ell_{q'}\subseteq\ell_{q}$ (see \eqref{eq:def:hbesov}).
Furthermore, for $1\leq p\leq p'\leq\infty$ there holds
$\hbesov_{pq}^s\subseteq\hbesov_{p'q}^{s'}$ when $s-s'=1/p-1/p'$
(recall that we consider only dimension one).

We will further abbreviate some a fixed $q$ (here for $q=1$)
$\hbesov^s_p:=\hbesov_{p1}^{s}$ and the scale--invariant
Besov spaces (for $q=1$) with $\hbesov_p:=\hbesov_{p}^{1/p}$.

In particular, $\hbesov_p\cap L_\infty$ are (Quasi--) Banach
algebras\footnote{The meaning of \eqref{eq:hbesov:algebra}
  has to be taken with care since homogeneous Besov spaces are
  \emph{equivalent classes} of 
  distributions modulo polynomials. The precise statement 
  can be found in \cite[Theorem 3.26]{Moussai2012}.
} (see here \cite[Remark 2 on p.148]{Peetre1976})
and we have for $s>0$ \cite[Lemma 1.5 on p.293]{Herz1968} and also \cite[Theorem 3.26]{Moussai2012}:
\begin{equation}
   \lVert fg\rVert_{\hbesov^s_{pq}}\leq c\left(\lVert
     f\rVert_{\hbesov^s_{pq}}\lVert g\rVert_{L_\infty}+\lVert
     g\rVert_{\hbesov^s_{pq}}\lVert f\rVert_{L_\infty}\right)
   \label{eq:hbesov:algebra}
\end{equation}

\paragraph{Trace--Class Results}
We have to consider paracommutators $\fourier{b}(\xi-\eta)A(\xi,\eta)$ with 
$A(\xi,\eta)=\fourier{a}(\eta)-\fourier{a}(\xi)$ where $b$ is given in
\eqref{eq:szego:b} and $a(z)=\phi(z/(2\alpha))/z$. We will follow the
notation in \cite{Janson1988}. The function $A(\xi,\eta)$ is a
uniformely bounded Schur multiplier  and vanishes on the diagonal, i.e., $A(\xi,\xi)=0$
(conditions A1 and A3($\infty$) in \cite{Janson1988}). The following
theorem holds:
\begin{mytheorem}[Thm. 8.1 in  \cite{Janson1988}]
  Let $T_b$ be a paracommutator in the form
  \eqref{eq:szego:def:paracomm} with
  $A(\xi,\eta)=\fourier{a}(\eta)-\fourier{a}(\xi)$
  and $a$ as defined above.
  For $1\leq p\leq\infty$ it holds:
  \begin{equation}
    \lVert T_b\rVert_{\schattenclass_p}\leq C\lVert b\rVert_{\hbesov_p}
  \end{equation}
  \label{thm:besov:traceclass}
\end{mytheorem}
We have to investigate paracommutators $T_b$ with symbols $b$ given by
equation \eqref{eq:szego:b}, i.e., which depend on the symbol $\sigma(\cdot,\omega)$
for a particular $\omega$ and on the interval length $\alpha$. Thus, which conditions on the
$x$--dependency of the symbols $\sigma(x,\omega)$ of $L_\sigma$ ensure
a sublinear scaling in $\alpha$ of $\lVert b\rVert_{\hbesov_p}$ ?
According to \eqref{eq:szego:b} this involves the behavior of Besov
norms with respect to pointwise multiplication with the smooth cutoff functions
$\phi_\alpha$ and compositions of the form  $\exp(i2\pi\nu\sigma)$. 
We will discuss both topics separately, yielding 
Lemma \ref{lemma:besov:pointwise} and Lemma \ref{lemma:besov:composition},
and combine the results at the end as Theorem \ref{thm:besov:traceclass:updated}.

\subsection{Pointwise Multiplications on Besov Spaces}
In this part we will explicitely investigate the behavior of the
homogeneous Besov norms with respect to pointwise multiplications. We
need conditions on a distribution $g$ which allow for the following limit:
\begin{equation}
  \frac{1}{\alpha} \lVert \phi(\cdot/\alpha)\cdot g\rVert_{\hbesov_1}
  \rightarrow 0\quad\text{for}\quad\alpha\rightarrow\infty
  \label{eq:szego:besov:mult:limit}
\end{equation}
where $\phi$ is a smooth cutoff function and, in general, $g$ is
\emph{not vanishing at infinity}. The reason for this is that
for $g$ we shall later use for example
$\exp(i2\pi\sigma(\cdot,\omega))$. Obviously, 
to obtain non--zero transmission capacity the symbol $\sigma(x,\omega)$,
somehow representing the ``channel power'' over time $x$,
can not vanish for $x\rightarrow\infty$.
Unfortunately, then - the algebra
property \eqref{eq:hbesov:algebra} is not helpful in this context.
\todo[inline]{not vanishing at infinity only for $p=\infty$}

Instead a Hoelder--type inequality for Besov spaces is required here
which  is known for inhomogeneous Besov spaces (see for example \cite{Runst1996}).
Using similar results for the homogeneous spaces in \cite{Sawada2003} we get the following theorem:
\begin{mylemma}
  Let $\phi\in C^\infty_c$ and $g\in\hbesov_\infty^{1}$. Then
  \begin{equation}
    \lim_{\alpha\rightarrow\infty}\frac{1}{\alpha} \lVert \phi(\cdot/\alpha)\cdot g\rVert_{\hbesov_1}=0
  \end{equation}
  \label{lemma:besov:pointwise}
\end{mylemma}
A well--known fact is that $g\in\hbesov_{\infty,1}^{1}$ if
and only if $g'\in\hbesov_{\infty,1}^{0}$. 
\begin{proof}
  For a given $1\leq r\leq\infty$ the Hoelder--dual
  exponent is defined by $1/r+1/\bar{r}=1$ and let us abbreviate then: 
  \begin{equation}
    [f,g]_{\hbesov_{r}^{a,b}}:=
    \lVert f\rVert_{\hbesov_{r}^a}\cdot\lVert g\rVert_{\hbesov_{\bar{r}}^{-b}}+
    \lVert f\rVert_{\hbesov_{r}^{-b}}\cdot\lVert g\rVert_{\hbesov_{\bar{r}}^a}
  \end{equation}
  We use Proposition and Remark 5 on p.11 in \cite{Sawada2003} for the case $p=q=s=1$.
  Let $1\leq r\leq\infty$ with  $\sigma>0$,
  $\theta\geq0$, $0<\delta\leq 1$ and $0\neq N\in\Naturals$. It holds:
  \begin{equation}
    \lVert f g\rVert_{\hbesov_1}\leq
    (N^2+1)[f,g]_{\hbesov_{r}^{\theta+1,\theta}}+
    2^{-N\delta}(N+1)\left(
      [f,g]_{\hbesov_{r}^{\sigma+1+\delta,\sigma}}+ 
      [f,g]_{\hbesov_{r}^{\sigma+1-\delta,\sigma}}
    \right)
  \end{equation}
  We will apply this now on \eqref{eq:szego:besov:mult:limit} where
  $f=\phi(\cdot/\alpha)$. From the scaling property \eqref{eq:hbesov:homogenity}
  we have that:
  \begin{equation}
    \frac{1}{\alpha}\lVert\phi_\alpha\rVert_{\hbesov_{r}^s}
    \leq c\alpha^{1/r-s-1}\lVert\phi\rVert_{\hbesov_{r}^s}
    =c\lVert\phi\rVert_{\hbesov_{r}^s}/\alpha^{s+1/\bar{r}}
    \label{eq:szego:besov:cutoff:limit}
  \end{equation}
  Intuitively, we would like to take $(r,\bar{r})=(1,\infty)$ to support that $g$
  is non--vanishing at infinity and \eqref{eq:szego:besov:cutoff:limit}
  implies then
  stricly--positive smoothness $s>0$
  such that \eqref{eq:szego:besov:mult:limit} is possible.
  But, from the theorem above we have to take into account all
  $s\in \{-\theta,-\sigma,\theta+1,\sigma+1\pm\delta\}$ where
  $\sigma>0$, $\theta\geq 0$ and $0<\delta\leq 1$. 
  Thus, the critical -- negative -- exponents for $1/\alpha$ 
  in \eqref{eq:szego:besov:cutoff:limit}
  are
  $s\in\{-\theta,-\sigma\}$ and therefore $(r,\bar{r})=(0,\infty)$ is not
  directly possible. More precisely, 
  the condition: 
  \begin{equation}
    \bar{r}< 1/\max(\theta,\sigma)
  \end{equation}
  is necessary such
  \eqref{eq:szego:besov:cutoff:limit}
  vanishes for $\alpha\rightarrow\infty$.
  
  Nevertheless, it possible to pose a dependency on $\alpha$ such that
  the limit can approached and, above a certain $\alpha_0$, all requirements of the theorem are
  fullfilled for each finite $\alpha\geq\alpha_0$.
  To this end we set $L=\log(\alpha)$ and choose $r=(L+1)/L>1$ and $\bar{r}=L+1$
  such that $(r,\bar{r})\rightarrow(1,\infty)$ for
  $\alpha\rightarrow\infty$. Next, we parametrize the smoothness
  parameters $\theta$, $\sigma$ and $\delta$. Fix some $\epsilon>1$. Then there
  exists $\alpha_0$ such that:
  \begin{equation}
    1\geq\delta=\theta=\sigma:=1/(L+1)-1/L^\epsilon>0\quad\text{for all}\quad\alpha\geq\alpha_0
  \end{equation}
  Summarizing:
  \begin{equation}
    (r,\bar{r},\theta,\sigma,\delta)\rightarrow(1,\infty,0,0,0)\quad\text{for}\quad\alpha\rightarrow\infty
  \end{equation}
\end{proof} 

\todo[inline]{Examples for 
  $\hbesov_{\infty,1}^{0}$ can be found in
  \cite{Sawada:2004b}. Further remarks are in 
  \cite{Sawada:2004c}
}

\subsection{Compositions on Besov Spaces $\besov_\infty^{s}$ for
  $s\geq 1$}
Here we discuss now how to handle the composition problem for
$\exp(i2\pi\nu\cdot)\circ\sigma(\cdot,\omega)$. We will use the
property that for each $\omega$ the function
$f:x\rightarrow\sigma(x,\omega)$ is bounded (since we already posed the
assumptions $\sigma\in C^3$ and $\sigma\in L_\infty$). This means that
$f\in\besov_\infty^s$ is in the inhomogeneous Besov space if
$f\in\hbesov_\infty^s$. Recall that in our case $s=1+\delta$ whereby
$\delta>0$ and the case $\delta=0$ is still open.

\begin{mylemma}
  Let $f:\Reals\rightarrow\Reals$ be real--valued and $f\in\besov_\infty^s$ for $s\geq 1$. Then there exists a
  constant $c>0$ depending on $f$ such that it holds:
  \begin{equation}
    \lVert \exp(i2\pi \nu f)\rVert_{\hbesov_\infty^{s}}\leq c(2\pi\nu)^s
  \end{equation}
  For $s=1$ the constant $c$ is of the form 
  $c=c'(1+\lVert f\rVert_{\besov_\infty^{1}})$ for some other constant $c'$
  (independent of $f$).
  \label{lemma:besov:composition}
\end{mylemma}
\begin{proof}
  For $s>1$ and $1\leq q\leq\infty$ we have from Theorem 4 in
  \cite{Bourdaud2002} that composition operator $T_G: f\rightarrow G\circ f$ fulfills
  $T_G(\besov_{\infty,q}^s)\subseteq \besov_{\infty,q}^s$ if and only if $G\in\besov_{\infty,q}^s$.
  More precily, there
  exists a continouos increasing function
  $\psi:\Reals_+\rightarrow\Reals_+$ such that it holds:
  \begin{equation}
    \lVert G\circ f\rVert_{\besov_{\infty,1}^s}\leq
    \lVert G\rVert_{\besov_{\infty,1}^s}\psi(\lVert f\rVert_{\besov_{\infty,1}^s})
    \label{eq:besov:composition:bourdaud}
  \end{equation}
  for all $f,G\in\besov_{\infty,1}^s$. 
  We have:
  \begin{equation}
    \begin{split}
      \lVert \exp(i2\pi \nu f)\rVert_{\hbesov_{\infty,1}^s}
      &\leq
      \lVert\cos(2\pi \nu f)\rVert_{\hbesov_{\infty,1}^s}+\lVert\sin(2\pi \nu f)\rVert_{\hbesov_{\infty,1}^s}\\
      &\overset{\eqref{eq:hbesov:homogenity}}{\leq} (2\pi\nu)^s(
      \lVert\cos(f)\rVert_{\besov_{\infty,1}^s}+\lVert\sin(f)\rVert_{\besov_{\infty,1}^s}) \\
      &\overset{\eqref{eq:besov:composition:bourdaud}}{\leq} (2\pi\nu)^s(
      \lVert\cos\rVert_{\besov_{\infty,1}^s}+\lVert\sin\rVert_{\besov_{\infty,1}^s})\psi(\lVert f\rVert_{\besov_{\infty,1}^s})
    \end{split}
  \end{equation}
  where we switched to the inhomogeneous Besov spaces since $\cos$ and
  $\sin$ are bounded. From Theorem $5$ in \cite{Bourdaud2002} it
  follows also that $\psi(x)=c'(1+x)$ for $s=1$.
\end{proof}
  
\subsection{Combining the Results}
Here we will now combine the results so far. The following theorem is
not the most general combination of the previous results. Instead we
have preferred for the moment a straightforward enumeration of the statements.
\begin{mytheorem} 
  Let $\sigma$ be a real--valued symbol with
  (i) $\sigma(\cdot,\omega)\in\besov_\infty^s$ for $s>1$ uniformely in
  $\omega$ and (ii) $\sigma(x,\cdot)\in L_1$ uniformely in $x$.
  Then it holds:
  \begin{equation}
    \lim_{\alpha\rightarrow\infty}Q_\alpha(\nu)/\alpha=0
  \end{equation}
  \label{thm:besov:traceclass:updated}.
\end{mytheorem}
Again, at this point it is open whether the same can be obtained for $s=1$.
\begin{proof}
  Recall that from \eqref{eq:szego:qalpha1} we have:
  \begin{equation}
    \begin{split}
      Q_\alpha(\nu)
      \leq
      c_1\lVert PTP\rVert_{\schattenclass_1}      
      +\lVert PT'P\rVert_{\schattenclass_1}+c_2\sqrt{\alpha(1+\log(\alpha))}
    \end{split}
    \label{eq:szego:qalpha2}
  \end{equation}
  and from \eqref{eq:szego:PTP}
  \begin{equation}
    \lVert PTP\rVert_{\schattenclass_1}\leq \int \lVert PT_\omega P\rVert_{\schattenclass_1}  d\omega\quad(\text{same for }T')
    \label{eq:szego:PTP:1}
  \end{equation}
  whereby the convergence of the integrals is ensured by sufficient
  decay of $\sigma(x,\omega)$ in $\omega$.
  For each $\omega$ the operators $T_\omega$ and $T'_\omega$ can be
  replaced by smoothed para--commutators $T_b$ of the form 
  \eqref{eq:szego:def:paracomm} with
  $A(\xi,\eta)=\fourier{a}(\eta)-\fourier{a}(\xi)$
  and $a(z)=\phi(z/(2\alpha))/z$ and symbol $b=b(\alpha,\sigma)$
  depending on $\alpha$ and $\sigma$ according to formula \eqref{eq:szego:b}.
  We apply now Lemma \ref{lemma:besov:pointwise} and Lemma
  \ref{lemma:besov:composition} to $T_b$ for each $\omega$. 
\end{proof}

\section{Conclusion}
A new approach to the capacity of time--continuous doubly--dispersive Gaussian channels
with periodic symbol has been established by proving a Szeg\"o asymptotic 
for certain pseudo--differential operators. The result holds once the
symbol $\sigma(x,\omega)$ of the channel correlation operator
$L_\sigma$ has sufficient decay in frequency $\omega$ and the
time--oscillation have finite $\besov_\infty^s$ Besov norm for $s>1$,
both is meant in a uniform sense.

\section*{Acknowledgment}
This work is supported by Deutsche Forschungsgemeinschaft (DFG) grant JU 2795/1-1.
The author would like to thank Holger Boche, Igor Bjelakovic,
Winfried Sickel and Jan Vybiral.

\bibliographystyle{IEEEtran}
\bibliography{library}

\end{document}